\documentclass[review]{siamart1116}

\usepackage[latin1]{inputenc}
\usepackage[english]{babel}
\usepackage[T1]{fontenc}
\usepackage[htt]{hyphenat}
\usepackage{amssymb}
\usepackage{stmaryrd} 
\usepackage{placeins}
\usepackage{amsmath,amsfonts}
\usepackage{graphicx}
\newcommand\fauxsc[1]{\fauxschelper#1 \relax\relax}
\def\fauxschelper#1 #2\relax{%
  \fauxschelphelp#1\relax\relax%
  \if\relax#2\relax\else\ \fauxschelper#2\relax\fi%
}
\def\Hscale{.85}\def\Vscale{.72}\def\Cscale{1.10}
\def\fauxschelphelp#1#2\relax{%
  \ifnum`#1>``\ifnum`#1<`\{\scalebox{\Hscale}[\Vscale]{\uppercase{#1}}\else%
    \scalebox{\Cscale}[1]{#1}\fi\else\scalebox{\Cscale}[1]{#1}\fi%
  \ifx\relax#2\relax\else\fauxschelphelp#2\relax\fi}
\usepackage{epstopdf}
\DeclareGraphicsExtensions{.pdf,.eps,.png,.jpg,.mps}
\usepackage{indentfirst}
\usepackage{dsfont}
\usepackage{color}
\usepackage{enumitem}
\usepackage{float}
\usepackage{pict2e}
\usepackage{soul}
\usepackage{array}
\newcolumntype{L}[1]{>{\raggedright\let\newline\\\arraybackslash\hspace{0pt}}m{#1}}
\newcolumntype{C}[1]{>{\centering\let\newline\\\arraybackslash\hspace{0pt}}m{#1}}
\newcolumntype{R}[1]{>{\raggedleft\let\newline\\\arraybackslash\hspace{0pt}}m{#1}}

\usepackage[OT2,OT1]{fontenc}
\newcommand\cyr
{
\renewcommand\rmdefault{wncyr}
\renewcommand\sfdefault{wncyss}
\renewcommand\encodingdefault{OT2}
\normalfont
\selectfont
}
\DeclareTextFontCommand{\textcyr}{\cyr}

\usepackage{url}

\newtheorem*{corollary2*}{Corollary~1 in Elbassioni \cite{Elbassioni2015}}
\newtheorem*{thm*}{Theorem~21 in Flum and Grohe \cite{Flum2004}}

\newtheorem{remark}{Remark}[section]

\usepackage{lipsum}
\usepackage{amsfonts}
\usepackage{graphicx}
\usepackage{epstopdf}
\usepackage{algorithmic}
\ifpdf
  \DeclareGraphicsExtensions{.eps,.pdf,.png,.jpg}
\else
  \DeclareGraphicsExtensions{.eps}
\fi

\newcommand{\TheTitle}{A general purpose algorithm for counting simple cycles and simple paths of any length} 
\newcommand{\RunningTitle}{Counting simple cycles and simple paths} 

\newcommand{\TheAuthors}{P.-L. Giscard, N. Kriege, and R. C. Wilson}

\headers{\RunningTitle}{\TheAuthors}

\title{{\TheTitle}\thanks{Submitted to the editors DATE.
\funding{P.-L.~Giscard acknowledges financial support from the Royal Commission for the Exhibition of 1851. N.~Kriege is supported by the German Science Foundation (DFG) within the Collaborative Research Center SFB 876 ``Providing Information by Resource-
Constrained Data Analysis'', project A6 ``Resource-efficient Graph Mining''.}}}

\author{
  Pierre-Louis Giscard\thanks{Department of Computer Science, University of York  (\email{pierre-louis.giscard@york.ac.uk}, \email{richard.wilson@york.ac.uk})}
  \and
  Nils Kriege\thanks{Department of Computer Science, TU Dortmund (\email{nils.kriege@tu-dortmund.de})}
  \and
  Richard C. Wilson\footnotemark[2]
}

\usepackage{amsopn}

\ifpdf
\hypersetup{
  pdftitle={\TheTitle},
  pdfauthor={\TheAuthors}
}
\fi

\usepackage{dsfont}

\begin{document}

\maketitle

\begin{abstract}
%
We describe a general purpose algorithm for counting simple cycles and simple paths of any length $\ell$ on a (weighted di)graph on $N$ vertices and $M$ edges, achieving a time complexity of 
$O\left(N+M+\big(\ell^\omega+\ell\Delta\big) |S_\ell|\right)$. In this expression, $|S_\ell|$ is the number of (weakly) connected induced subgraphs of $G$ on at most $\ell$ vertices, $\Delta$ is the maximum degree of any vertex and $\omega$ is the exponent of matrix multiplication. 
We compare the algorithm complexity both theoretically and experimentally with most of the existing algorithms for the same task. These comparisons show that the algorithm described here is the best general purpose algorithm for the class of graphs where $(\ell^{\omega-1}\Delta^{-1}+1) |S_\ell|\leq |\text{Cycle}_\ell|$, with $|\text{Cycle}_\ell|$ the total number of simple cycles of length at most $\ell$, including backtracks and self-loops. On Erd\H{o}s-R\'enyi random graphs, we find empirically that this happens when the edge probability is larger than circa $4/N$. In addition, we show that some real-world networks also belong to this class.
Finally, the algorithm permits the enumeration of simple cycles and simple paths on networks where vertices are labeled from an alphabet on $n$ letters with a time complexity of $O\left(N+M+\big(n^\ell\ell^\omega+\ell\Delta\big) |S_\ell|\right)$. 
A Matlab implementation of the algorithm proposed here is available for download.
\end{abstract}



\normalsize

\section{Introduction}
Counting Hamiltonian cycles and, more generally, all simple cycles passing through a given vertex is a $\#$P-complete problem \cite{Valiant1979,Burgisser1997}. The same classification holds for the problem of counting simple paths with fixed endpoints. 
Unsurprisingly, the best existing algorithms for counting such cycles have time complexities $O\big(2^N\text{poly}(N)\big)$, which scales exponentially with the number $N$ of vertices on the graph. Under the exponential time hypothesis \cite{Impagliazzo2001}, this exponential scaling is, in principle, the best possible.\\[-.7em]

Although evaluating the time complexity of an algorithm in the worst case scenario is of paramount importance for the classification of algorithmic performance, it is of little relevance to applications which differ significantly from this scenario. This is precisely the case when counting or enumerating simple cycles or simple paths. Real-world networks, be they from sociology, biology or chemistry, are typically very sparse. At the opposite, the worst case scenarios for this task|the complete graphs|are dense and counting or finding cycles and paths of any kind on them presents no interest, in particular since everything is already known analytically. An algorithm counting simple cycles and paths that is especially tailored for sparse graphs is therefore highly desirable. In particular, we expect the graph sparsity, or some quantity related to it, to be a relevant parameter when qualifying the complexity of such an algorithm.\\[-.7em]

%
In addition to these considerations, we observe that one rarely needs to count \textit{all} simple cycles or paths. Rather it is typically sufficient to count only those whose length does not exceed some maximum value $\ell$, usually much smaller than the graph size $N$.
Yet, even with these restrictions, the problem of counting simple cycles or simple paths is known to be difficult:\\
\begin{thm*}
Counting simple cycles and simple paths of length $\ell$ on both directed and undirected graphs, parameterized by $\ell$, is $\#$W[1]-complete.\\
\end{thm*}
The complexity classes $\#$W[$t$], $t\geq 1$, introduced by Flum and Grohe, are relevant for parameterized counting problems corresponding to the classes of the W-hierarchy~\cite{Downey1999} which, in turn, qualify the difficulty of parametrized decision problems according to the type of circuits needed to determine them. Importantly, the class $\#$W[1] is believed to \textit{strictly} contain the class $\#$W[0] of all fixed-parameter tractable (FPT) counting problems.
We recall that a counting problem $P$ with input $x$ is said to be \textit{fixed-parameter tractable} if there is a computable function $f$ of the parameter $k$, a constant $c$ and an algorithm solving $P$ in $f(k) \,\text{poly}(|x|)$ steps. In this expression, $|x|$ designates the size of the input \cite{Flum2004,Grohe1999}. For the sake of simplicity, an algorithm achieving a $f(k) \,\text{poly}(|x|)$ time complexity will be said to be FPT.\\[1em] 


In this work, we describe a novel general purpose algorithm for the task of counting simple cycles and simple paths of fixed length $\ell$ and determine its time complexity:\\  
\begin{theorem}[Algorithm for cycle and path counting]
\label{ThmMain}
Let $G=(V,E)$ be a graph, possibly directed, on $N$ vertices and $M$ edges. Let $|S_\ell|$ be the number of connected induced subgraphs of $G$ on at most $\ell$ vertices. Let $\Delta$ be the maximum degree of any vertex on $G$ or, if $G$ is directed, let $\Delta$ be the maximum degree of any vertex on the undirected version of $G$.
Then all the simple cycles of length up to $\ell$ on $G$ can be counted in time
\begin{equation*}
O\left(N+M+\big(\ell^\omega+\ell\Delta\big) |S_\ell|\right),
\end{equation*}
and $O\big(N+M\big)$ space. 
The same complexity is achieved when counting the simple paths of length up to $\ell$ or the simple cycles/paths with fixed endpoints of length up to $\ell$.\\ 
\end{theorem}

The important result of Theorem~\ref{ThmMain} is that the time complexity of the general purpose algorithm presented here scales as $\text{poly}(\ell)|S_\ell|$. In comparison, we show in Section~\ref{Comparisons} that the time complexities of all other general purpose algorithms scale either with $N^\ell$, which is always larger than $|S_\ell|$ unless the graph is complete, or with the number $|\text{Cycle}_\ell|$ of simple cycles of length at most $\ell$ on the graph.\footnote{There is one exception to this observation: by extending an approach of Merris to count Hamiltonian cycles \cite{Merris1983}, we show in Section~\ref{Comparisons} that all simple cycles can be counted with a time complexity scaling as $\ell\,t_{\text{imm}}(\ell)|S_\ell|$, where $t_{\text{imm}}(\ell)$ is \textit{exponential} in $\ell$. Hence, this extension is still not competitive with the algorithm presented here.} From these observations, we expect that the algorithm presented here be the best available for graphs with less connected induced subgraphs than simple cycles,\footnote{Note in this context, backtracks, that is bidirected edges, count as simple cycles. Furthermore the orientation of the cycles counts as well. Thus, for exemple, the complete graph on three vertices with no self-loops, $K_3$, has two simple cycles of length 3 and three of length 2.} something we both confirm and precise in Section~\ref{Comparisons}. While deciding a-priori if a graph obeys this condition is difficult, we will see in Section~\ref{Experiments} that it is true for several real-world networks and most Erd\H{o}s-R\'enyi random graphs.\\[-.7em] 

\begin{remark}\label{FPTremark}
The algorithm presented here is FPT for the problem of counting simple cycles or simple paths of length $\ell$, parameterized by $\ell$, for the class of graphs where the number of connected induced subgraphs on at most $\ell$ vertices fulfils $|S_\ell|=O\big(f(\ell)\,\text{poly}(N)\big)$, with $f$ a computable function. This class is the class of bounded degree graphs, for which the existence of an FPT algorithm is not surprising. Indeed,  the number of simple cycles / paths of length $\ell$ is upper bounded by the number of walks of this length, which is at most $\Delta^\ell N=f(\ell)\,\text{poly}(N)$. In fact, on bounded degree graphs, even a direct search of the simple cycles achieves the same complexity and constitutes a FPT algorithm.
%
\end{remark}

~
\vspace{1mm}

We prove Theorem~\ref{ThmMain} in Section~\ref{Thm1Proof} using the analytical framework outlined in Section~\ref{SectionAnalytical} below. Following the proof of Theorem~\ref{ThmMain}, we compare in Section~\ref{Comparisons} the performance of the algorithm presented here with the time complexities achieved by existing algorithms for the same task. These fall in five families: i) combinatorial sieves; ii) cycle counts by zeon-algebras; iii) cycle counts from combinations of immanants; iv) special identities for short length cycles on undirected graphs; and v) counting via enumeration. 
In Section~\ref{Experiments}, we present numerical experiments validating the results of Theorem~\ref{ThmMain} and the comparisons of Section~\ref{Comparisons}. We then demonstrate the performance of the algorithm on real-world networks. The \textit{Matlab} implementation and data sets used for these experiments is available for download at \texttt{https://www-users.cs.york.ac.uk/$\sim$plg508/}. We conclude in Section~\ref{LabelledGraphs} with an extension of the algorithm for the enumeration of simple cycles on graphs with few labels.\\[1em]

\section{Analytical framework}\label{SectionAnalytical}
\subsection{Counting simple cycles}\label{CycleCount}
Our algorithm is based on a recent result from algebraic combinatorics relating the numbers of walks and of simple cycles on any (directed) graph. This result provides an explicit formula for the  ordinary generating function of the number $\gamma(\ell)$ of simple cycles of length $\ell$ multiplied $\ell$ \cite{Giscard2016b}
\begin{equation}\label{IntermediateComb}
\sum_{\ell} \ell \gamma(\ell) z^\ell = \sum_{H\prec_{\text{conn}} G} \mathrm{Tr}\left((z \mathsf{A}_H)^{|H|}(\mathsf{I}-z\mathsf{A}_H)^{|N(H)|}\right),
\end{equation}
where the sum runs over all weakly connected induced subgraphs $H$ of $G$. Recall that a directed graph is said to be weakly connected if and only if its undirected version is connected. Thus such subgraphs can be found by an algorithm for finding connected induced subgraphs running over the undirected version $G_\text{undir.}$ of $G$. In this expression, $|H|$ designates the number of vertices of the subgraph $H$ and $|N(H)|$ is the number of neighbours of $H$ in $G$. A neighbour of $H$ in $G$ is a vertex $v$ of $G$ which is not in $H$ and such that there exists at least one edge, possibly directed, from $v$ to a vertex of $H$ or from a vertex of $H$ to $v$. Finally, $\mathsf{A}_H$ is the adjacency matrix of $H$.
From the formula of Eq.~(\ref{IntermediateComb}) for the generating function of $\ell \gamma(\ell)$, we obtain $\gamma(\ell)$ analytically as  
\begin{equation}\label{EqforGammaL}
\gamma(\ell)=\frac{(-1)^\ell}{\ell}\sum_{H\prec_{\text{conn}} G}\binom{|N(H)|}{\ell-|H|}(-1)^{|H|}\,\mathrm{Tr}\left(\mathsf{A}_H^\ell\right).
\end{equation}
This explicit result forms the basis of the algorithm proposed here: counting the simple cycles can be achieved by evaluating Eq.~(\ref{EqforGammaL}).\\[-.7em]

Any algorithmic implementation of Eq.~(\ref{EqforGammaL}) can be easily compared with the best existing combinatorial sieve for counting simple cycles, that of Bax and Franklin \cite{bax1996finite,bax1994algorithms}, by observing that Eq.~(\ref{EqforGammaL}) involves a sum over the weakly \emph{connected} induced subgraphs of the graph. In contrast, Bax and Franklin's algorithm evaluates a formula involving a sum over \emph{all} the induced subgraphs, including the non-connected ones. Remarkably, in the worst case scenario|the complete graph|every induced subgraph is connected, making it look like both algorithms should have a comparable complexity.
On any other graph however, there are far more induced subgraphs than connected induced subgraphs. This crucial difference means that evaluating Eq.~(\ref{EqforGammaL}) \textit{must} yield a significant speed-up as compared to Bax and Franklin's algorithm. This argument is made rigorous by the proof of Theorem~\ref{ThmMain}, which we present in Section~\ref{Thm1Proof}, and see also Section~\ref{BaxCompare}. Before we proceed to this proof however, we present extensions of Eq.~(\ref{EqforGammaL}) for counting simple paths and simple cycles with fixed end points.\\

\subsection{Counting simple paths and simple cycles visiting a fixed vertex}
To find all the simple paths, we rely once more on a recent result from algebraic combinatorics according to which the ordinary generating function of the number $\pi_{i\to j}(\ell)$ of simple paths of length $\ell$ from vertex $i$ to $j$ is the $ij$-entry of \cite{Giscard2016b}
\begin{equation}\label{GenPi}
\sum_\ell \pi_{i\to j}(\ell) \,z^\ell = \Big(\,\sum_{H\prec_{\text{conn}}G} (z \mathsf{A}|_H)^{|H|-1}(\mathsf{I}-z\mathsf{A}|_H)^{|N(H)|}\,\Big)_{ij}\,.
\end{equation}
This expression employs the same notation as that presented in Section~\ref{CycleCount} with the exception of $\mathsf{A}|_H$, which represents the adjacency matrix of $G$ restricted to the connected induced subgraph $H$. That is, $$
\big(\mathsf{A}|_H\big)_{ij}=\begin{cases}
\mathsf{A}_{ij},&\text{if }i,j\in H,\\
0,&\text{otherwise,}
\end{cases}\quad i,j\in G.
$$
This construction allows one to formally write the sum of the various terms on the right hand side of Eq.~(\ref{GenPi}). Most importantly, from a computational point of view, multiplying by $\mathsf{A}_H$ or $\mathsf{A}|_H$ has the same time complexity $O(|H|^\omega)$. The number $\pi_{i\to j}(\ell)$ then follows analytically as
\begin{equation}\label{EqforPiL}
\pi_{i\to j}(\ell)=(-1)^{\ell+1}\sum_{H\prec_{\text{conn}} G\atop i,j\in H}\binom{|N(H)|}{\ell+1-|H|}(-1)^{|H|}\,\left(\mathsf{A}|_H^\ell\right)_{ij},
\end{equation}
where the sum now runs over all weakly connected induced subgraphs of $G$ containing both $i$ and $j$. 
With the notation introduced here, we may also extend the result of Eq.~(\ref{EqforGammaL}) to count only those simple cycles passing through any specified vertex $i$ with 
\begin{equation}\label{EqforGammaiL}
\gamma_{i}(\ell)=(-1)^\ell\sum_{H\prec_{\text{conn}} G\atop i\in H}\binom{|N(H)|}{\ell-|H|}(-1)^{|H|}\,\left(\mathsf{A}|_H^\ell\right)_{ii}.
\end{equation}
In particular, we verify immediately that 
\begin{equation}\label{EqGammaiLGammaL}
\gamma(\ell)=\frac{1}{\ell}\sum_{i=1}^N \gamma_i(\ell),
\end{equation}
as expected.

\newpage
\section{Proof of Theorem~\ref{ThmMain}}\label{Thm1Proof}
\subsection{The Algorithm: Evaluating Equation~(\ref{EqforGammaL})}
The algorithm consists simply in evaluating Eq.~(\ref{EqforGammaL}), (\ref{EqforPiL}) or (\ref{EqforGammaiL}), depending on what one wants to count.
Given the similar structures of these equations it is readily apparent that of Eqs.~(\ref{EqforGammaL}), (\ref{EqforPiL}) and (\ref{EqforGammaiL}), it is Eq.~(\ref{EqforGammaL}) that necessitates the greatest computational effort to be evaluated. This observation is best encapsulated by Eq.~(\ref{EqGammaiLGammaL}). For the sake of simplicity and to concretely illustrate our arguments, we will thus determine the time complexity explicitly only in the costliest situation: evaluating Eq.~(\ref{EqforGammaL}).\\ 

We first remark that the binomial coefficient appearing in Eq.~(\ref{EqforGammaL}) is non-zero if and only if $|H|\leq \ell\leq |N(H)|+|H|$. Thus, only those weakly connected induced subgraphs $H$ of $G$ on $|H|\leq \ell$ vertices contribute a term of Eq.~(\ref{EqforGammaL}) when calculating $\gamma(\ell)$. Equivalently, all the weakly connected induced subgraphs of $G$ such that $|H|\leq \ell$ provide all the terms needed to calculate all $\gamma(k),~k\leq \ell$. 
Therefore, for an algorithm based on Eq.~(\ref{EqforGammaL}) to count the simple cycles, it is sufficient for it to find weakly connected induced subgraphs of bounded size. This observation leads to the following result:
\begin{lemma}\label{BasicLemma}
Let $t\big(|S_{\ell}|\big)$ be the time complexity of finding all the weakly connected induced subgraphs of $G$ on at most $\ell$ vertices. 
Then the time complexity of determining the number $\gamma(k)$ of simple cycles of length $k$ for all $k\leq \ell$ is 
$$O\left(t\big(|S_{\ell}|\big)+(\ell^\omega+\ell\Delta) |S_{\ell}|\right).$$
\end{lemma}

\begin{proof}
This is straightforward from Eq.~(\ref{EqforGammaL}). First, all the weakly connected induced subgraphs on $k\leq \ell$ vertices must be found, costing $O\left(t\big(|S_{\ell}|\big)\right)$ time, by definition. Second, the terms of Eq.~(\ref{EqforGammaL}) must be evaluated, of which there are $|S_\ell|$ in total. Each term involves counting: i) the number of neighbours $|N(H)|$ of a subgraph $H$ on the graph; and ii) the walks of length $\ell$ on this subgraph, its size being $|H|\leq \ell$. The first step thus costs at most $|H|\Delta= O(\ell\Delta)$ time and the second step requires $|H|^\omega= O(\ell^\omega)$ time. 
\end{proof}
~\\[-1em]

\subsection{Time complexity of finding the connected induced subgraphs}
As we have seen it is necessary and sufficient to find all the weakly connected induced subgraphs of size at most $\ell$ to count simple cycles of length up to $\ell$. This can be done using the standard reverse search algorithm for finding connected induced subgraphs introduced by Avis and Fukuda \cite{Avis1996}, running on $G_{\text{undir.}}$ the undirected version of the graph $G$. 
The total running time of this algorithm in our case is \cite{Uehara1999, Elbassioni2015}
\begin{equation}\label{revsearchcost}
t\big(|S_{\ell}|\big)=O\left(N+M+\sum_{k=1}^{\ell-1}|S_{=k}|+\ell^2|S_{=\ell}|\right)=O\Big(N+M+\ell^2|S_\ell|\Big),
\end{equation}
where $M=|E|$ is the number of edges in $G_{\text{undir.}}$ and $|S_{=k}|$ designates the number of connected induced subgraphs on \emph{exactly} $k$ vertices in $G_{\text{undir.}}$. Furthermore this algorithm uses $O(N+M)$ space. We also remark that thanks to reverse search, the algorithm for counting simple cycles can be parallelised: indeed, the contribution of each connected induced subgraph to Eq.~(\ref{EqforGammaL}) can be calculated independently of the other subgraphs. Now combining Eq.~(\ref{revsearchcost}) with Lemma~\ref{BasicLemma} and noting that $\omega\geq 2$ concludes the proof of Theorem~\ref{ThmMain}.\\[-.7em]

\begin{remark}
The time complexity of the reverse search algorithm for finding the connected induced subgraphs was recently improved upon by Karakashian \text{et al.} \cite{Karakashian2013} and then further by K. Elbassioni \cite{Elbassioni2015}. Elbassioni describes a polynomial delay algorithm for this task yielding the following time complexity:\\[-1ex]
\begin{corollary2*}\label{ElbaLemma}
Finding all the connected induced subgraphs of size $k\leq \ell$ in a graph with maximum degree $\Delta$ can be done in 
$$
O\left(\ell^2\min\{(N-\ell),\ell\Delta\}(\log N+\Delta+\log \ell)\,|S_\ell|\right),
$$
total time.\\[-1ex] 
\end{corollary2*}
\noindent Elbassioni also presents an algorithm with a slightly worse time complexity, but ensuring a $O(N+M)$ space complexity, see Theorem~1 in \cite{Elbassioni2015}. Unfortunately, implementations of these recent algorithms have not yet been produced.\\
\end{remark}

\subsection{Understanding the time complexity}
In the worst case scenario, that is the complete graphs $K_N$, Theorem~\ref{ThmMain} implies that the time complexity for counting \textit{all} the simple cycles using the algorithm proposed here is $O(2^N N^\omega)$ since all induced subgraphs are connected, i.e. $|S_N|=2^N$. This is marginally better than the complexities reported in \cite{karp1982dynamic,bax1996finite,bax1994algorithms,schott2011complexity}. 
However, it is the performance of our algorithm on non-complete graphs that we want to highlight. To this end, it is helpful to recast the time complexity of the algorithm in terms of simple graph parameters.\\[-.7em]
%

We can do so by using an upper bound on the number $|S_{k}|$ of connected induced subgraphs on $k$ vertices that involves the maximum degree of any vertex. This result is due to Uehara:\\[-1ex]
\begin{lemma}[Uehara \cite{Uehara1999}]\label{UehLemma}
 Let $\Delta$ be the maximum degree of the undirected version $G_\text{undir.}$ of $G$. Then the number of connected induced subgraphs on exactly $k$ vertices in $G_\text{undir.}$ is bounded by
$$
|S_{=k}|\leq N\frac{(e\Delta)^k}{(\Delta-1)k^{2}},
$$
with $e$ the base of the natural logarithm. It follows that $$|S_\ell|=\sum_{k\leq \ell}|S_{=k}|=O\left(N\frac{\Delta^\ell}{(\Delta-1)\ell^2}\right).$$
\end{lemma}
Furthermore, on a graph with maximum degree $\Delta$, there are at most $M\leq N\Delta$ edges, so that, by Theorem~\ref{ThmMain}, the time-complexity of counting all the simple cycles of length $k\leq \ell$ is upper bounded by
\begin{subequations}
\begin{align}
O\left(N(\Delta+1)+(\ell^\omega+\ell\Delta)N\frac{\Delta^\ell}{(\Delta-1)\ell^2}\right)&=O\left(N\Delta+N(\ell^{-1}\Delta+\ell^{\omega-2})\Delta^{\ell-1}\right),\label{bound}\\
&\sim O\big(N\ell^{-1}\Delta^\ell\big),\label{estimate}
\end{align}
\end{subequations}
where we used that $\Delta/(\Delta-1)\leq 2$ as soon as the graph has a connected component with at least 3 vertices.\\[-.7em]  

The bound on $|S_\ell|$ obtained from Uehara's work is typically \emph{very far from tight}, especially on graphs that are far from regular, such as scale-free networks. Consequently, the time complexity predicted by Eq.~(\ref{estimate}) is typically much larger than that observed in numerical experiments. However, Eq.~(\ref{estimate}) simplifies the analysis of the time complexity of the algorithm, which will help us compare it with other algorithms for the same task. Observe also that we now easily verify the claim of Remark~\ref{FPTremark} that the algorithm is FPT on bounded degree graphs. In fact, on such graphs $\Delta=O(1)$, consequently the time complexity scales as $N$, that is the algorithm is fixed parameter \emph{linear}.\\[2em]

\section{Detailed comparisons with existing algorithms}\label{Comparisons}
\subsection{Sieve methods}\label{BaxCompare}
Bax and Bax and Franklin authored two articles detailing the use of combinatorial sieves to count simple cycles \cite{bax1996finite,bax1994algorithms}, which extend previous results by Karp \cite{karp1982dynamic} for counting Hamiltonian cycles. Similar techniques had previously been expounded by Khomenko and Golovko \cite{Khomenko1972,khomenko1978problem} and more recently by Perepechko and Voropaev \cite{Perepechko2009,PVRussianArticle}.\\[-.7em]

All these combinatorial sieves produce the simple cycles via sums over all the induced subgraphs of a graph, i.e. including the non-connected ones. 
There are $\tbinom{N}{\ell}$ such subgraphs of size $\ell$ on a graph on $N$ vertices. Assuming $\ell$ is fixed and much smaller than $N$, the number of subgraphs is $\Omega(N^\ell/\ell!)$. Consequently, counting all simple cycles of length up to $\ell$ using these sieves takes at least $\Omega(N^\ell/\ell!)$ time. If $\Delta$ is sub-linear in $N$, this time complexity is much larger than that achieved by the algorithm presented here, which takes at most $O(N\Delta^\ell/\ell)$ time. In other terms, Eq.~(\ref{EqforGammaL}), which only involves the \emph{connected} induced subgraphs, yields a significant speed-up. If instead $\Delta=\alpha N$, $0< \alpha\leq1$, the algorithm presented here is still $1/\alpha^{\ell}$ faster than other combinatorial sieves.\footnote{In addition, we have empirically observed that the time complexity of the algorithm proposed here scales with an effective parameter $\Delta_\text{eff}\ll\Delta$. What determines $\Delta_\text{eff}$ remains unclear.} 
\\ 

\subsection{Zeons algebras}
An algorithm for counting simple cycles based on zeon algebras has been proposed by Schott and Staples in \cite{schott2011complexity}. The algorithm relies on the observation that if one attaches a formal variable $\xi_e$ to each edge $e$ of the graph, such that any two such variable commute and $\xi_e^2=0$, then the corresponding labeled adjacency matrix $(\mathsf{A}_\xi)_{ij}:=\xi_{ij}\mathsf{A}_{ij}$ generates only simple cycles. In other terms, $\mathrm{Tr}(A^\ell_\xi)$ is the number of simple cycles of length $\ell$ on $G$. Unfortunately, this method requires formal matrix multiplications and cannot be implemented fully numerically.\\[-.7em]

Schott and Staples proved that the average time taken by this algorithm to count simple cycles of length $\ell$ is 
$O\big(N^4(1+q)^N\big)$ where $q\geq\ell N\Delta/(N^2-\ell)$ \cite{schott2011complexity}. In the typical situation where $\Delta,\ell\ll N$, this cost is therefore at least $O\big(N^4 e^{\ell \Delta})$. This is exponential in both $\ell$ and $\Delta$ and scales as the fourth power of $N$, hence always much larger than the $O(N\Delta^\ell/\ell)$ bound obtained earlier.\\

\subsection{Counting using immanants}
In 1983, R. Merris discovered an exact formula for counting the Hamiltonian cycles of a graph from a sum over at most $N$ of its immanants \cite{Merris1983,Merris2005}. On noting that any simple cycle is Hamiltonian on an unique connected induced subgraph of the graph, Merris' formula 
is easily extended to count all simple cycles of length up to $\ell$ via a sum over the $|S_\ell|$ connected induced subgraphs of size at most $\ell$.  In this sum, each term is itself a sum over at most $\ell$ immanants. Therefore, evaluating the formula takes $O\big(t(|S_{\ell}|)+t_\text{imm}(\ell) \ell |S_\ell|\big)$ time, with $t(|S_{\ell}|)$ and $t_\text{imm}(\ell)$ the times taken to find the connected induced subgraphs on at most $\ell$ vertices and to calculate the required immanants of $\ell\times \ell$ matrices, respectively.\\[-.7em]

In the same spirit, G. Cash described in 2007 an approach for counting simple cycles by solving a system of equations involving selected immanantal polynomials of the graph \cite{cash2007number}.
For length $\ell$ simple cycles, Cash's approach stems from the solution of a system involving $p(\ell)-p(\ell-1)$ equations, where $p(\ell)$ is the number of integer partitions of $\ell$. 
This number grows as $O(e^{x\,\sqrt{\ell}}\ell^{-3/2})$ with $x=\pi\sqrt{2/3}\sim 2.6$ and consequently solving the system takes $O(e^{7.7\,\sqrt{\ell}}\ell^{-9/2})$ time. Since the immanantal polynomials of the graph take $O\big(t_\text{imm}(N)\big)$ time to calculate, the cost of Cash's approach is $O\big(t_\text{imm}(N) e^{7.7 \sqrt{\ell}}\ell^{-9/2}\big)$.\\[-.7em]

Most importantly, we see that the time complexities of both methods are primarily influenced by the time taken to calculate the required immanants. Unfortunately, these are difficult to obtain. First, as recognized by Cash, they require computing the matrix of irreducible representations of the symmetric group $\mathcal{S}_x$, a very costly task for large $x$. Second, while the determinant of an $x\times x$ matrix requires only $O(x^3)$ time, the second immanant $d_2$ already costs $O(x^c)$ with $3<c\leq 4$ and computing the last immanant, the permanent, is itself a $\#$P-complete problem \cite{Valiant1979b}. The permanent is required by both Merris' and Cash's approaches, meaning that, assuming the exponential time hypothesis, $t_\text{imm}(x)$ grows exponentially in $x$. Comparing with Theorem~\ref{ThmMain}, we observe that neither approach can compete with the algorithm proposed here.\\[-.7em]

In the same vein, Giscard, Rochet and Wilson showed that simple cycles can be counted via a combination of permanents and determinants summed over the set of induced subgraphs of a graph \cite{Giscard2016c}, of which there are $\Theta(N^\ell)$. This particular approach only requires the computation of two immanants per subgraph thereby bypassing the need for computing the matrices of irreducible representations of large symmetric groups, yet requires all induced subgraphs to be considered and thus takes a prohibitive $O\big(N^\ell t_{\text{imm}}(\ell)\big)$ time.\footnote{Even if a reduction to connected induced subgraphs can be devised for this method, which would yield a $O(t(S_\ell) + t_{\text{imm}}(\ell)|S_\ell|)$ time complexity, it would only marginally improve upon Merris' approach and would still be worse than that of the algorithm presented here.}\\




\subsection{Counting short simple cycles on undirected graphs}
When only short simple cycles on \textit{undirected graphs} are of interest, these may be counted via a set of special identities involving the adjacency matrix. This approach was pioneered by Harary and Manvel in the 1970s and has remained popular ever since \cite{Harary1971, Alon1997, Chang2003,Movarraei2016}. In particular, Alon, Yuster and Zwick presented an algorithm for evaluating these identities up to $\ell=7$ in $O(N^\omega)$ time and $O(N^2)$ space \cite{Alon1997}. This cost grows for longer cycles, being $O(N^{\omega+1})$ when $\ell=8$, and then $O(N^{\lfloor\ell/2\rfloor}\log N)$ when $\ell=9,~10$. To the best of our knowledge, no special identity for counting $\ell>10$ cycles has been found. There is little doubt that these exist however. Yet, given that the formula for $\ell=10$ already involves 160 terms  \cite{Perepechko2009}, any such identity would be extremely cumbersome.\\[-.7em] 

From this discussion, we conclude that if the graph is undirected, only short simple cycles of length $\ell\leq 7$ are desired and $N\leq \Delta^{\frac{\ell}{\omega-1}}$, then we expect Alon, Yuster and Zwick's approach (AYZ) to be faster than the algorithm presented here. In practice, we find AYZ to be faster in many cases where this condition is not met, presumably because of differing constant factors hidden by the $O(.)$ notation. These, in turn, might stem from the fact that AYZ is not a general purpose algorithm which relies on specific, optimised, ways of counting the simple cycles. Since the time complexity of AYZ scales with $N^\omega$ however, for any family of graphs where $\Delta=o(N)$, there is a graph size above which the algorithm presented here is faster than AYZ.\\[-.7em]

Finally, we note that the space required for running AYZ scales as $O(N^2)$ rather than $O(N+M)$, the former being much larger than the latter on sparse graphs.  We found this to be AYZ main limitation in practice,\footnote{That is, beyond the 
fact that AYZ is limited to $\ell=7$ on undirected graphs.} barring us from 
making computations on networks with over 12,000 nodes. This memory cost 
is unavoidable since AYZ necessitates the computation of powers of the 
adjacency matrix $\mathsf{A}$ of the graph, which quickly become dense 
even on large sparse graphs. Recall in particular, that $\mathsf{A}^x$ is full 
for $x$ larger than the graph diameter.\\


\subsection{Counting simple cycles via enumeration}
Enumerating the simple cycles or simple paths of a graph, that is producing their vertex sequences, is much more time consuming than simply counting them. The best general purpose algorithm for this task is still Johnson's 1975 landmark algorithm \cite{Johnson1975,Mateti1976}, which achieves a time complexity of $O\left((N+M)\,(|\text{Cycle}_N|+1)\right)\sim O\left(N\Delta |\text{Cycle}_N|\right)$. In this expression, $|\text{Cycle}_N|$ is the total number of simple cycles (or of simple paths) on $G$, including backtracks, that is simple cycles of length 2.
This result was recently improved on undirected graphs to $O(N\,(|\text{Cycle}_N|+1)+M)$, a scaling which is optimal for this task \cite{Birmele2013}.\\[-.7em]

In the worst case scenario, i.e. on the complete graph $K_N$,  $|\text{Cycle}_N|=O(N!)$, that is enumerating all simple cycles takes factorial time. For this reason, counting simple cycles via enumeration has often been deemed greatly inefficient,  in particular in comparison with the ``only" exponential cost $O\big(2^N\text{poly}(N)\big)$ achieved by the algorithm presented here as well as other approaches \cite{bax1994algorithms}. This conclusion follows from a peculiarity of dense graphs however and for sparse graphs it is not so.\\[-.7em]

Indeed, evaluating Eq.~(\ref{EqforGammaL}) to count all the simple cycles on a graph costs $O(N^\omega|S_N|)$. It follows that if $N^\omega |S_N|\geq N\Delta|\text{Cycle}_N|$, then Johnson's algorithm and its variants can count all the simple cycles of a graph via enumeration faster than any combinatorial sieve, including the one presented here. When counting simple cycles of fixed maximum length $\ell$, Johnson's algorithm takes $O(N+M+(\ell+\ell\Delta)|\text{Cycle}_\ell|)$ time, $|\text{Cycle}_\ell|$ being the total number of simple cycles of length up to $\ell$. This means in order for the algorithm presented here to be faster than Johnson's the following must hold
$$\left(\frac{\ell^{\omega-1}}{\Delta}+1\right) |S_\ell|\leq |\text{Cycle}_\ell|.$$
Unfortunately, it is very difficult to estimate the ratio $|S_\ell|/ |\text{Cycle}_\ell|$ in a preprocessing stage so as to decide which algorithm to use. Furthermore, the problem of characterising graphs for which the number of connected induced subgraphs is larger than the number of simple cycles is, to the best of our knowledge, an open mathematical question beyond the scope of this work. Rigorously, we may only conclude that for any number $N$ of vertices, there must be a critical density above which the algorithm presented here will be faster than Johnson's. 
We undertake an empirical study of this density in Section~\ref{Experiments} on Erd\H{o}s-R\'{e}nyi random graphs and show it to be so small that the resulting graphs are disconnected with very high probability. 
In addition, we also study two families of real-world networks exemplifying the interplay between connected induced subgraphs and simple cycles.\\
%


\subsection{Relation to subgraph counting algorithms}
Given a pattern graph $H$ and a graph $G$, the \emph{subgraph counting problem} 
is to determine the number of (induced) subgraphs of $G$ that are isomorphic to 
$H$. This problem is well studied in undirected graphs and generalizes the problem 
of counting cycles. Therefore, we briefly summarize results on the subgraph 
counting problem, in particular, when  parameterized by the size of the pattern 
graph $k=|V(H)|$.\\[-.7em]

When $G$ is a planar graph on $N$ vertices the problem can be solved in time 
$N 2^{O(k)}$~\cite{Dorn2010} improving the seminal FPT algorithm by Eppstein, 
which achieves $N k^{O(k)}$ time~\cite{Eppstein1999}.
Ne\v{s}et\v{r}il and Ossona de Mendez introduced \emph{classes of graphs with 
bounded expansion}, which include the planar graphs as well as the graphs of 
bounded degree. For these classes they have shown that the number of satisfying 
assignments of a Boolean query with a fixed number of free variables can be 
counted in time linear in $N$~\cite[Theorem 18.9]{Nesetril2012}, which 
solves the subgraph counting problem for patterns of a fixed size as a special case.
An improved algorithm tailored to counting subgraphs was proposed by Demaine et
al.~\cite{Demaine2014} and achieves a time complexity of $O(6^k t^k k^2 N)$,
where $t$ is the height of a tree-depth decomposition of $G$. Again the approach
yields a linear time FPT algorithm in graph classes of bounded expansion when
parametrised by $k$.
The running times of the above mentioned algorithms typically hide enormous 
constants and, to the best of our knowledge, have for this reason not been applied 
in practice.\\[-.7em]

Technically related to our work is the method introduced by Amini et 
al.~\cite{Amini2012} to count subgraphs by homomorphisms using a combinatorial 
sieve. Their approach can be seen as a generalisation of the standard sieve methods for 
counting cycles to arbitrary graphs, but does not overcome the drawbacks 
regarding running time discussed in Section~\ref{BaxCompare}.\\[2em]

\section{Experiments}\label{Experiments}
In this section we present numerical evidence for the performance of a \textit{Matlab} implementation of the algorithm presented in this work. Note that this implementation incorporates a preprocessing stage which removes all sources, sinks and isolated vertices of the graph.
We compare it with both Johnson's and AYZ algorithms since, following Section~\ref{Comparisons}, these are the only competitive algorithms for counting simple cycles. 
For the former we use Howbert's freely available \textit{Matlab} implementation \cite{Howbert2011}, while for the latter we wrote a \textit{Matlab} code, available for download. 
All the calculations reported here have been made on a MacBook Pro laptop with 3.1 GHz Intel Core i7 processor and 8 GB of RAM running \textit{Matlab} R2016a.\\[-.5em]

\subsection{Erd\H{o}s-R\'{e}nyi random graphs}
We begin by considering undirected Erd\H{o}s-R\'{e}nyi random graphs $\text{ER}(N,p)$. These random graphs are determined by two parameters: the number $N$ of vertices and the probability $p$ that any one undirected edge in the graph exists (with the exception of self-loops). The expected number of edges in $\text{ER}(N,p)$ is $pN(N-1)/2$ so that the expected graph sparsity equals $p$.\\[-.7em] 

\subsubsection{Comparison with Johnson's algorithm}
We undertook the comparison on two ranges of parameters: small graphs $5\leq N\leq 30$, on which we compared the times taken by both algorithms to count all the simple cycles; and on large graphs $N\geq 1,000$, for which we counted simple cycles of length up to 5 only.\\[-.7em]

On small graphs, for each value of $N$ from 5 to 22 as well as for $N=25$ and $30$, we determined the critical value $p_{\text{critical}}(N)$ of $p$ below which Johnson's algorithm is the fastest by incrementing $p$ from 0 to 1 by steps of $10^{-2}$ at $N$ fixed. For each value of $p$, we ran both algorithms 20 times and compared the averaged time taken, except for $N=25$ and $30$ where we ran the algorithms only twice per value of $p$.
The results are shown on Figure~\ref{ERRegion}.
Empirically we observe that for small graphs, $N\leq 30$, Johnson's algorithm is faster to count all simple cycles whenever $p\leq p_{\text{critical}}(N)\simeq 4.3/N$. Equivalently, this means that Johnson's algorithm can be expected to be faster than the algorithm presented here whenever the average degree is close to 4 or smaller.\\[-.7em]

On large graphs $p_{\text{critical}}(N)$ falls further, being seemingly less than $1/N$ when counting simple cycles of length up to 5 on Erd\H{o}s-R\'{e}nyi graphs with 20,000 vertices. In any case, we remark that for $N\gg 1$ and  $p< \log(N)/N$, $\text{ER}(N,p)$ is known to be almost surely disconnected. Given that we observed that $p_{\text{critical}}(N)=O(1/N)$, it seems that for Johnson's algorithm to be the fastest, an Erd\H{o}s-R\'{e}nyi random graph must be so sparse as to be disconnected in many small components.\\   
\begin{figure}[t!]
\begin{center}
\includegraphics[width=1\linewidth]{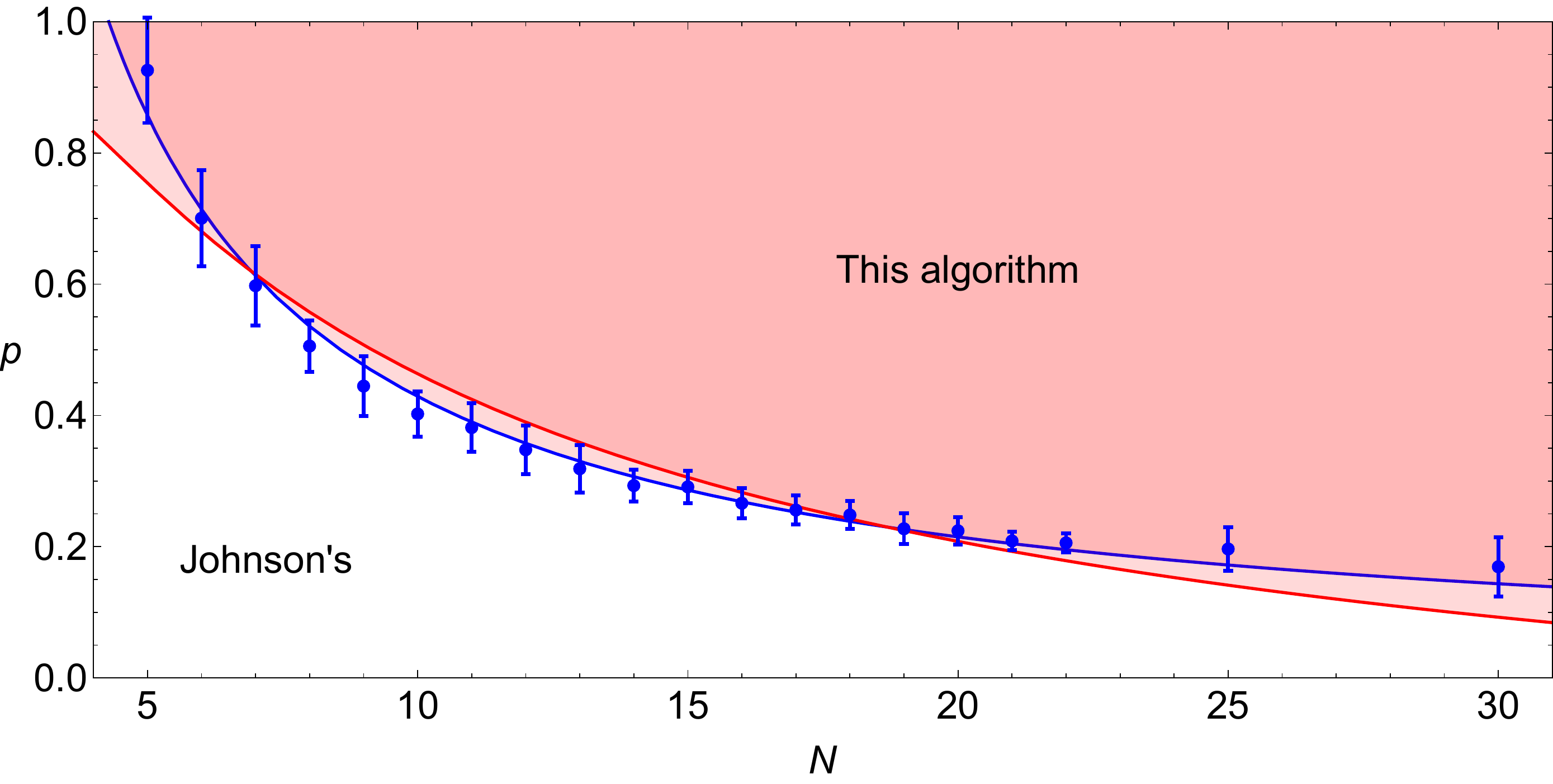}
\end{center}
\caption{\label{ERRegion} Experimental comparison of Johnson's algorithm with the algorithm presented here. The blue dots are the observed values of the critical edge probability below which Johnson's algorithm is the fastest. 
The red line shows the best fit of the blue data points of the form $p_{\text{critical}}(N)=a+b\log(N)/N$, which is
$p_{\text{critical}}(N)=-0.27+ 3.17\log(N)/N$.
The blue line shows the best fit of the blue data points of the form $p_{\text{critical}}(N)=a+b/N$, which is
$p_{\text{critical}}(N)=8.5\times10^{-4} + 4.28/N$. Since the latter fit is clearly  better than the former, we retain $p_{\text{critical}}(N)\simeq 4.3/N$ as model when discussing $p_{\text{critical}}$ in the text.}
\end{figure}

\subsubsection{Comparison with AYZ}
The algorithm of Alon, Yuster and Zwick is almost always the fastest to count simple cycles of length only up to 7 on undirected Erd\H{o}s-R\'enyi random graphs. Indeed, we find that as soon as the graph is \textit{denser} than  $p_{\text{critical}}\approx 1.23/N$, for $N\lesssim 10,000$, then AYZ is faster than the algorithm presented here. The value of $p_{\text{critical}}$ slowly increases with larger values of $N$ and is around $p_{\text{critical}}\simeq 1.3/N$ for $N\sim 12,000$. 
Unfortunately, the memory consumption of AYZ barres us from directly studying the performances of both algorithms on larger graphs. Rather, an extrapolation of the increase of $p_{\text{critical}}$ suggests that it crosses $\log(N)/N$ when $N$ is well over $10^6$. This is widely beyond what can be reached by AYZ, and so large as to render the algorithm presented here prohibitively slow. We must thus conclude that AYZ remains the fastest algorithm on Erd\H{o}s-R\'enyi random graphs for counting simple cycles of length 7 or less.\\[2em]

\subsection{Real-world networks}
We applied the algorithm presented in this work as well as those of AYZ and Johnson to compare their performances on three real-world networks, two of which are undirected and one is both weighted and directed:\\[-.5em]
\begin{description}
 \item [\fauxsc{Actors}:] This network represents collaborations between movie actors and 
was generated and analysed by Barabási and Albert~\cite{Barabasi1999}.
Each actor constitutes a vertex and an edge represents that the two actors were 
cast together in the same movie.\\[-1em]
 \item [\fauxsc{Infectious}:] A network representation of the face-to-face contacts 
 between visitors of the exhibition \textit{Infectious: Stay Away} held in Dublin in 
 2009~\cite{Isella2011}.
 Each edge corresponds to face-to-face interaction lasting for at least 20 seconds.\\[-.5em]
\end{description}
All data sets were obtained from the \textsc{Konect} website~\cite{Kunegis2013}, where
further information on the data sets is available. 
The networks are undirected and have parallel edges. In order to systematically 
study the effect of sparsity, we generated several instances of each network by 
deleting edges with multiplicity below a given threshold as follows.
Starting with the graph with all edges present, we successively removed all edges with 
multiplicity $1, 2, \dots$. A new instance is created whenever at 
least 80 edge were removed from the previous instance. This results in a sequence of graphs with decreasing density progressively retaining only the most important edges. The sequence based on the \textsc{Actors} network comprises 31 graphs, while the sequence based on the \textsc{Infectious} network comprises 8 graphs.\\[-.7em]
 

\begin{description}
 \item [\fauxsc{Wikielections}:] A weighted directed network representing the votes of Wikipedia users during elections to
adminship \cite{StanfordWiki}. Each user corresponds to a vertex, and a directed edge from user $u_1$ to user $u_2$ exists if and only if $u_1$ voted during the election of $u_2$. This edge is given a $+1$ weight if user $u_1$ supported $u_2$ candidacy and $-1$ otherwise.
No pruning of the edges was operated on this network.\\[-.5em]
\end{description}
\noindent The Wikielections network has  8289 vertices and 12915 directed edges. It is a scale-free graph with maximum out-degree $\Delta^{\text{out}}=266$ and maximum in-degree $\Delta^{\text{in}}=191$. Being directed, the Wikielections network cannot be studied with  AYZ, which is limited to undirected graphs, nor can it be studied with Howbert's implementation of Johnson's algorithm.\\[-.5em]

In all cases, in order to accurately describe the performances of the algorithm presented here, we provide, for each graph, the time $\tau_\ell$ it takes for counting all simple cycles of length up to $\ell$, as well as the parameter governing the scaling of this time with $\ell$. Indeed, while $\tau_\ell$ is upper bounded by $N\Delta^\ell/\ell$ as per Eq.~(\ref{estimate}), empirically, we find it to scale as $\tau_\ell\propto\Delta_{\text{eff}}^\ell$ with $\Delta_{\text{eff}}<\Delta$. This effective scaling parameter is determined numerically by fitting $\tau_\ell$ with $a\times\Delta_{\text{eff}}^\ell+b$, where $a$ and $b$ are fitted constants. Surprisingly, we did not find any relation between $\Delta_{\text{eff}}$ and the maximum, mean, or median of the vertex degrees. What determines its value in practice remains unclear.

\begin{table*}[h!]
\begin{center}
\renewcommand{\arraystretch}{1}
\begin{tabular}{|C{2cm}|c|C{1.5cm}|C{.9cm}|C{.8cm}|C{1.9cm}|L{4cm}|}
\hline
\hline  
\textsc{Instance \#}&$N$&$M$&$\Delta$&$\Delta_\text{eff}$&\textsc{Time (sec.)}&\hspace{.7cm}\textsc{Simple cycles}\\
\hline
31&45&96&7&1.4&A:~$4\times10^{-2}$ J:~$3\times 10^{-2}$&$\ell=10:~$0, 48, 32, 48, 48, 44, 16, 0, 0, 0 \\
\hline
30&90&260&12&1.6&\text{A}:~1.34\hspace{10mm} J:~4.06&$\ell=10:~$0, 130, 202, 652, 2044, 5876, 14046, 25700, 33148, 29820 \\
\hline
29&143&428&17&3.3&\text{A}:~19\hspace{10mm} J:~76&$\ell=10:~$0, 214, 356, 1328, 4946, 18608, 62038, 175710, 398864, 705874\\
\hline
28& 179&588&24&4.6&A:~248\hspace{10mm} J:~667&$\ell=10:~$0, 294, 566, 2564, 11830, 56066, 246604, 970674, 3284880, 9284612\\
\hline
27&125&748&26&5.4&A:~970\hspace{10mm} J:~2660&$\ell=10:~$0, 374, 798, 4110, 22332, 125084, 665030, 3246496, 14068582, 52877616\\
\hline
26&257&914&30&5.8&A:~2829\hspace{10mm} J:~8349&$\ell=10:~$0, 457, 1018, 5726, 34724, 218028, 1310046, 7326752, 37074200, 166360444\\
\hline
25&310&1118&36&7.3&A:~12301\hspace{10mm}J:~40124&$\ell=10:~$0, 559, 1294, 7986, 53828, 377298, 2538470, 16045588, 92969672, 485843893\\
\hline
24&376&1414&47&9&A:~20992\hspace{10mm}J:~$>10^5$&$\ell=9:~$0, 707, 1728, 11686, 85300, 650344, 4744026, 32672232, 207557400\\
\hline
23&423&1610&49&9.6&A:~5766\hspace{10mm}J:~$>3\times 10^4$&$\ell=8:~$0, 805, 2058, 15008, 118748, 980604, 7827540, 59395940\\
\hline
22&470&1854&51&11.8&A:~$1.02\times10^4$\hspace{10mm}J:~$>5\times 10^4$&$\ell=8:~$0, 927, 2476, 18674, 154346, 1333982, 11215982, 90027620\\
\hline
17&863&3894&75&11.8&A:~$2.8\times10^4$\hspace{10mm}J:~$-$\hspace{10mm}AYZ:~0.45&$\ell=7:~$0, 1947, 6178, 61640, 688510, 8187720, 96547224\\
\hline
12&1911&10428&119&22&A:~$5.7\times10^4$\hspace{10mm}J:~$-$\hspace{10mm}AYZ:~0.91&$\ell=6:~$0, 5214, 22060, 330498, 5625464, 105644852\\
\hline
7&6085&48916&238&24&A:~$7.3\times10^4$\hspace{10mm}J:~$-$\hspace{10mm}\hspace{10mm}AYZ:~4.8&$\ell=5:~$0, 24458, 181724, 5127548, 169365078 \\
\hline
4&19199&235964&609&$\sim70$&A:~$4\times 10^5$\hspace{10mm}J:~$-$\hspace{10mm}AYZ: OOM&$\ell=4:~$ 0, 117982, 1608856, 103794848\\
\hline
1:~Full graph&382219&30076166&3956&$-$&A:~$-$\hspace{10mm}J:~$-$\hspace{10mm}AYZ: OOM&$\ell=4:$~$-$\\
 \hline 
 \hline
\end{tabular}
\caption{\label{ActorsResults}Counting simple cycles on some graphs of the \textsc{Actors} data set. The time taken by the algorithm presented in this work is labelled by "A", while "J" refers to the time taken by Johnson's algorithm.  We report the time taken by AYZ only when simple cycles of length 7 or less are counted. Because of memory limitations, we could not run AYZ on graphs 1 to 4, which we designated by "OOM" for "out of memory".
}
\end{center}
\end{table*}

\FloatBarrier

We now turn to the \textsc{Infectious} family of graphs. Contrary to the \textsc{Actors} set of graphs, we found Johnson's algorithm to run faster on this data set than the algorithm presented here.

\begin{table}[h!]
\begin{center}
\renewcommand{\arraystretch}{1.1}
\begin{tabular}{|C{2cm}|c|c|c|c|C{2cm}|L{5cm}|}
\hline
\hline  
\textsc{Instance $\#$}&$N$&$M$&$\Delta$&$\Delta_\text{eff}$&\textsc{Time (sec.)}&\hspace{1.2cm}\textsc{Simple cycles}\\
\hline
8&14&7&1&1&A:~$1.7\times 10^{-3}$\hspace{10mm} J:~$1.2\times10^{-3}$&$\ell= 10:~$0, 7, 0, 0, 0, 0, 0, 0, 0, 0\\
\hline
7&29&30&2&1&A:~$4.7\times 10^{-3}$\hspace{10mm} J:~$4.3\times10^{-3}$&$\ell= 10:~$0, 15, 0, 0, 0, 0, 0, 0, 0, 0\\
\hline
6&42&50&2&1&A:~$8.3\times 10^{-3}$\hspace{10mm} J:~$1.3\times10^{-2}$&$\ell= 10:~$0, 25, 4, 0, 0, 0, 0, 0, 0, 0\\
\hline
5&91&116&4&1&A:~$2.1\times 10^{-2}$\hspace{10mm} J:~$2.2\times10^{-2}$&$\ell= 10:~$0, 58, 12, 2, 0, 0, 0, 0, 0, 0\\
\hline
4&236&394&4&1.01&A:~$3.3\times 10^{-2}$\hspace{10mm} J:~$1.7\times10^{-1}$&$\ell= 10:~$0, 197, 94, 60, 16, 0, 0, 0, 2, 2\\
\hline
3&337&964&15&3.8&A:~$220$\hspace{10mm} J:~$33.8$&$\ell= 10:~$0, 482, 572, 1340, 3552, 9490, 23504, 50900, 92630, 143620\\
\hline
2&368&1760&24&6.4&A:~$2.3\times 10^4$\hspace{10mm} J:~$2.2\times 10^4$&$\ell= 10:~$0, 880, 2322, 11506, 65356, 391646, 2391434, 14585954, 87432978, 509475403\\
\hline
1:\,Full graph&410&5530&50&16.8&A:~$4.39\times10^4$\hspace{10mm}J:~$1.8\times 10^4$\hspace{10mm}AYZ:~0.4&$\ell= 7:~$0, 2765, 14228, 162574, 2142470, 30356160, 446411676\\
 \hline 
 \hline
\end{tabular}
\caption{\label{InfectiousResults}Counting simple cycles on the graphs of the \textsc{Infectious} data set. The time taken by the algorithm presented in this work labelled by "A", while "J" refers to the time taken by Johnson's algorithm. We report the time taken by AYZ only when simple cycles of length 7 or less are counted. 
}
\vspace{-5mm}
\end{center}
\end{table}

\FloatBarrier

~\\

Finally, we turn to the \textsc{Wikielections} network which, as indicated earlier,  is both directed and signed. The sign of a simple cycle being the product of the signs of its edges, we propose to demonstrate the algorithm capabilities by finding the numbers $p_\ell$ and $n_\ell$ of positive and negative simple cycles of length $\ell\leq 6$, respectively.
Indeed, since it is sufficient to run the algorithm twice to obtain both $p_\ell$ and $n_\ell$. More precisely, running the algorithm once on the signed network yields $p_\ell-n_\ell$, while running it on its unsigned version provides $p_\ell+n_\ell$.
\begin{table}[h!]
\begin{center}
\renewcommand{\arraystretch}{1.2}
\begin{tabular}{|C{2.5cm}|c|c|c|c|c|c|}
\hline
\hline  
\textsc{Length} &1&2&3&4&5&6\\
\hline
\textsc{Time taken (sec.)} &$6\times10^{-3}$&
$9\times10^{-2}$&2.2&84&2981&$1.04\times10^5$\\
\hline
\textsc{Positive simple cycles} &6& 337&1683& 16369&182657&2170663\\
\hline 
\textsc{Negative simple cycles} &5&12&253&3323&46792&663136\\
\hline
\textsc{Total} &11&349&1936&19692&229449&2833799\\
\hline
\hline
\end{tabular}
\caption{\label{WikiResults} Number of positive and negative simple cycles of length $\ell$ up to 6 on the Wikielections directed network and the time taken to count them. Here $\Delta_\text{eff} \simeq35$.
}
\end{center}
\end{table}
\FloatBarrier

%

\section{Finding labelled simple cycles and simple paths}\label{LabelledGraphs}
In some applications, such as chemoinformatics, counting the simple cycles or simple paths of a network is not sufficient. Rather, the vertices of the network may be labelled and it is then necessary to find all sequences of labels corresponding to simple cycles/paths on the network. If there are as many different labels as vertices, that is each vertex has its own label, then this task is best addressed by Johnson's algorithm discussed earlier \cite{Johnson1975}.\\ 

In typical applications however, the number of labels is much smaller than the number of vertices. For example, on a network representing a molecule where vertices are atoms and edges are bonds, vertex labels represent the various atomic species, e.g. carbon, hydrogen, nitrogen etc. There is less than 10 such species in the vast majority of organic molecules in available data sets. In addition, in the standard representation of molecules, hydrogen vertices are omitted altogether and the label of carbon atoms is put to the default value 1 (that is no label), further reducing the number of different labels.\\  

Finding all simple cycles/paths label sequences in such situations can be done with Eq.~(\ref{EqforGammaL}), (\ref{EqforPiL}) or (\ref{EqforGammaiL}) with the following time complexity:\\
\begin{theorem}
Let $G=(V,E)$ be a graph, possibly directed, on $N$ vertices. Let $\Delta$ be the maximum degree of any vertex on $G$ or, if $G$ is directed, let $\Delta$ be the maximum degree of any vertex on the undirected version of $G$. Finally, let $n$ be the number of different labels attached to graph vertices. 
Then all the label sequences of simple cycles of length up to $\ell$ on $G$ can be found in time
\begin{equation}
O\left(N+M+(n^\ell \ell^\omega+\ell\Delta)|S_{\ell}|\right)
\end{equation}
and $O\big(N(\Delta+1)\big)$ space.
The same complexities are achieved to find the label sequences of all simple paths up to length $\ell$ or of simple cycles/paths with fixed endpoints up to length $\ell$.\\ 
\end{theorem}

\begin{proof}
In the presence of labels, Eqs.~(\ref{EqforGammaL}), (\ref{EqforPiL}) and (\ref{EqforGammaiL}) continue to be valid and provide the label sequences of the simple cycles/paths upon replacing the adjacency matrix $\mathsf{A}$ by the labeled adjacency matrix $\mathsf{W}$, defined by
$$
\mathsf{W}_{ij}:= \mathsf{A}_{ij}\, w_{L(i)L(j)},
$$
where $L(i)$ is the label of vertex $i$ and $w$ is a formal variable. For example, a nitrogen-oxygen bond in a molecular graph would appear as $w_{NO}$ in $\mathsf{W}$. The time complexity of evaluating Eq.~(\ref{EqforGammaL}), (\ref{EqforPiL}) or (\ref{EqforGammaiL}) remains unchanged except for the cost of calculating the traces $\operatorname{Tr}(\mathsf{W}^\ell_{H})$. These can be obtained through matrix multiplications. Since the entries of $\mathsf{W}^{\ell-1}_{H}$ are sums of label sequences of walks of length $\ell-1$ on the subgraph $H$ and since there are at most $n^{\ell-1}$ such sequences, evaluating the trace costs at most $|H|^{\omega}n^{\ell-1} = O(\ell^\omega n^{\ell})$ time. Replacing $\ell^\omega$ with this cost in Theorem~\ref{ThmMain} yields the result. 
\end{proof}
~\\

\section{Conclusion}
We have presented a novel general purpose algorithm for counting simple cycles and simple paths of any length $\ell$ on any graph, including directed and weighted ones. The time complexity of this algorithm scales with the number $|S_\ell|$ of weakly connected induced subgraphs on at most $\ell$ vertices, making it the best general purpose algorithm whenever $ (\ell^{\omega-1}\Delta^{-1}+1) |S_\ell|\leq |\text{Cycle}_\ell|$. In this expression $|\text{Cycle}_\ell|$ is the total number of simple cycles of length up to $\ell$, including self-loops and backtracks.
Empirically, we found that this happens on Erd\H{o}s-R\'{e}nyi random graphs when the edge-probability exceeds circa $4/N$, as well as on some real-world networks, such as those in the \textsc{Actors} family of graphs.\\[-.7em] 

If the network under study is undirected and if counting simple cycles of length up to 7 is sufficient, then the algorithm of Alon, Yuster and Zwick is still by far the fastest. Furthermore, while we can predict that there must a graph size such that AYZ becomes slower than the algorithm presented here, on Erd\H{o}s-R\'enyi random graphs this size seems to be much beyond what we can reach. Indeed, we could not run AYZ on graphs with more than $N\gtrsim 12,000$ vertices owing to its important memory consumption. While this number can likely be increased with an optimised implementation of AYZ in conjunction with more memory, it is unlikely to get substantially larger as the memory usage of AYZ scales with $N^2$. In contrast, we could run the algorithm presented here on networks with over 300,000 vertices without running out of memory.\\[-.7em] 

Finally, even though the algorithm presented here is the best general purpose algorithm on the class of graphs with $(\ell^{\omega-1}\Delta^{-1}+1) |S_\ell|\leq |\text{Cycle}_\ell|$, the time necessary to count simple cycles of length e.g. up to 10 can be prohibitively large on large networks. Instead, the algorithm is best used in conjunction with Monte Carlo methods. We demonstrate this procedure in a separate publication \cite{Giscard2016d}, where we use it in a sociological context to obtain the ratios of negative to positive simple cycles of length up to 20 on several real-world directed signed networks with up to 130,000+ vertices.

\bibliographystyle{plain}

\end{document}